\documentclass[11pt]{article}
\usepackage{preamble}

\title{A note on quantum lower bounds for local search via congestion and expansion\footnote{Alphabetical author ordering.  Research  supported  by the US National Science Foundation under CAREER grant CCF-2238372.}}
\author{Simina Br\^anzei\footnote{Purdue University. E-mail: simina.branzei@gmail.com.}
	\and 
	Nicholas J. Recker\footnote{Purdue University. E-mail: nrecker@umich.edu.}
}

\date{} 

\begin{document}

\maketitle 

\begin{abstract}
We consider the quantum query complexity of local search as a function of graph geometry.  Given a graph $G = (V,E)$ with $n$ vertices and black box access to a function $f : V \to \mathbb{R}$, the goal is find a vertex $v$ that is a local minimum, i.e. with $f(v) \leq f(u)$ for all $(u,v) \in E$, using as few oracle queries as possible.

We show that  the quantum query complexity of local search on $G$ is $\Omega\bigl( \frac{n^{\frac{3}{4}}}{\sqrt{g}} \bigr)$, where $g$ is the vertex congestion of the graph. For a $\beta$-expander with maximum degree $\Delta$, this implies a lower bound of $ \Omega\bigl(\frac{\sqrt{\beta} \; n^{\frac{1}{4}}}{\sqrt{\Delta} \; \log{n}} \bigr)$. We obtain these bounds by applying the strong weighted adversary method to a construction from  \cite{branzei2024sharp}.

As a corollary, on constant degree expanders, we derive a lower bound of $\Omega\bigl(\frac{n^{\frac{1}{4}}}{ \sqrt{\log{n}}} \bigr)$.  This improves upon the best prior quantum lower bound of $\Omega\bigl( \frac{n^{\frac{1}{8}}}{\log{n}}\bigr) $  by \cite{santha2004quantum}.
In contrast to the classical setting,  a gap remains in the quantum case between our lower bound and the best-known upper bound of  $O\bigl( n^{\frac{1}{3}} \bigr)$ for such graphs.
%
\end{abstract}

\section{Introduction}

We analyze the quantum query complexity of local search as a function of the geometry of the underlying space. Specifically, we are given  a graph $G=(V,E)$ and oracle access to a function $f : V \to R$. The objective is to identify a vertex $v$ that represents a local minimum, meaning $f(v) \leq f(u)$  for every edge $(u,v) \in E$, while minimizing the number of vertices queried.

The  randomized query complexity of local search was first studied by \cite{aldous1983minimization}, who showed a lower bound of $\Omega(2^{n/2-o(n)})$ for the Boolean hypercube $\{0,1\}^n$. \cite{Aaronson06} obtained improved randomized lower bounds for the Boolean hypercube and $d$-dimensional grid by designing a classical variant of the relational adversary method from quantum computing.   \cite{Aaronson06} also showed quantum lower bounds, in particular proving that the quantum query complexity of local search on the Boolean hypercube $\{0,1\}^n$ is  $\Omega\left(2^{n/4}/n \right)$. A quantum upper bound of  $O\left(2^{n/3}n^{1/6} \right)$ was also given in \cite{Aaronson06}.

\cite{zhang2009tight} further improved these results, showing that  the quantum query complexity on the Boolean hypercube $\{0,1\}^n$ is $\Theta\left( 2^{n/3} n^{1/6}\right)$ and on the $d$-dimensional grid $[n]^d$ is $\Theta\left( n^{d/3} \right)$ for $d \geq 6$. \cite{sun2009quantum} analyzed the low dimensional regime for the grid, showing that quantum query complexity on the 2-dimensional grid $[n]^2$ is $\Omega \left( {n^{1/2 - \delta } } \right)$, and on the 3-dimensional grid $[n]^3$ is $\Omega (n^{1- \delta } )$, for any fixed $\delta \ge 0$.

\cite{santha2004quantum} showed a quantum lower bound as a function of separation number: for every graph $G$ on $n$ vertices, the quantum query complexity of local search on $G$ is $\Omega\left(\sqrt[8]{\frac{s}{\Delta}} /\log{n}\right)$, where $s$ is the separation number of the graph and $\Delta$ its maximum degree. \cite{Verhoeven06}  showed a general upper bound of $O(\sqrt{\Delta})+O(\sqrt[4]{\gamma})\cdot\sqrt[4]{n}\log\log n$, where $\gamma$ is the genus of the graph.

Our contribution is to show quantum lower bounds for every graph $G$ as a function of  features such as vertex congestion and expansion. We obtain these results via the strong weighted adversary method applied to a construction from \cite{branzei2024sharp}.



\section{Model}

Let $G = (V,E)$ be a connected undirected graph and $f : V \to \mathbb{R}$ a function defined on the vertices.  
A vertex $v \in V$  is a local minimum if $f(v) \leq f(u)$ for every edge $(u,v) \in E$.
We  write $V = [n] = \{1, \ldots, n\}$.

Given as input a graph $G$ and oracle access to function $f$, the local search problem is to find a local minimum of $f$ on $G$ using as few queries as possible. {Each query is of the form: ``Given a vertex $v$, what is $f(v)$?''}.

\paragraph{Quantum query model.}

In the quantum query model, an algorithm has a working state $|\phi \rangle$ of the form $|\phi \rangle = \sum_{v, a, z} \alpha_{v, a, z} | v, a, z \rangle $, where $v$ is the label of a vertex in $G$, $a$ is a string representing the answer register,  $z$ is a string representing the workplace register, and $\alpha_{v, a, z}$ is a complex amplitude. 
The amplitudes  satisfy the constraint $\sum_{v, a, z} |\alpha_{v, a, z}|^2 = 1$.

Starting from an arbitrary fixed initial state $|\phi_0\rangle$, the algorithm works as an alternating sequence of \emph{query} and \emph{algorithm} steps. If the current state is $| \phi \rangle = \sum_{v,a,z} \alpha_{v,a,z} \left | v, a, z\rangle \right. $, a query transforms it  as follows
\begin{align} \label{eq:def_U_x}
    \sum_{v,a,z} \alpha_{v,a,z} \left | v, a, z\rangle \right. \rightarrow \sum_{v,a,z} \alpha_{v,a,z} | v, a \oplus f(v), z\rangle,
\end{align}
where $\oplus$ denotes the bitwise exclusive OR operation.
An algorithm step multiplies the vector of $\alpha_{v,a,z}$'s by a unitary matrix that does not depend on $f$.
Thus a $T$-query quantum query algorithm is a sequence of operations
\begin{align}
    U_0 \rightarrow Q \rightarrow U_1 \rightarrow Q \rightarrow \ldots \rightarrow  U_{T-1} \rightarrow Q \rightarrow U_T, 
\end{align}
where $Q$ is the oracle gate defined in \eqref{eq:def_U_x} and $U_0, U_1, \ldots, U_T$ are arbitrary  unitary operations that are independent of the input function $f$.

Let $\mathcal{L}_f$ denote the set of local minima of $f$. The algorithm is said to succeed if at the end it gives a correct answer with probability at least $2/3$, that is:
\begin{align} 
\sum_{v,a,z: v \in \mathcal{L}_f} |\alpha_{v,a,z}|^2 \geq 2/3\,. \notag 
\end{align}
The bounded error quantum query complexity $Q(G)$ is defined as the minimum number of queries used by a quantum algorithm that succeeds on every input function $f$ for  the graph $G$.

\paragraph{Congestion.}
Let $\cP = \{P^{u,v}\}_{u,v \in V}$  be an all-pairs set of paths in $G$, where $P^{u,v}$ is a path from $u$ to $v$.  {For convenience, we assume $P^{u,u} = (u)$  for all $u \in V$; our results will  hold even if  $P^{u,u} = \emptyset$.}

For a path $Q = (v_1, \ldots, v_s)$ in $G$, let $c^Q_v$ be the number of times a vertex $v \in V$ appears in $Q$ and $c^Q_e$ the number of times an edge $e \in E$ appears in $Q$.
The \emph{vertex congestion} of the set of paths $\cP$ is $\max_{v \in V} \sum_{Q \in \cP} c_v^Q$, while the  \emph{edge congestion} of $\cP$ is $\max_{e \in E} \sum_{Q \in \cP} c_e^Q$.

The \emph{vertex congestion of $G$} is the smallest integer $g$ for which the graph has an all-pairs set of paths $\mathcal{P}$ with vertex  congestion $g$. Clearly, $g \geq n$ since each vertex belongs to at least $n$ paths in $\mathcal{P}$ and $g \leq n^2$ since each vertex appears at most once on each path and there are $n^2$ paths in $\mathcal{P}$. 
 The \emph{edge congestion} $g_e$ is similarly defined, but with respect to the edge congestion of a  set of paths $\mathcal{P}$.

\paragraph{Expansion.}
For each set of vertices $S \subseteq V$, the edges with one endpoint in $S$ and another in $V \setminus S$ are called \emph{cut edges} and denoted 
\begin{align}
E(S, V \setminus S) = \bigl \{ (u,v) \in E \mid  u \in S, v \not\in S \bigr \} \,. \notag 
\end{align} 
The graph is a $\beta$-expander if $|E(S, V \setminus S)| \geq \beta \cdot |S|$, for all $S \subseteq V$ with $0 < |S| \leq n/2$ (see, e.g. \cite{alon04}).
The graph is \emph{$\Delta$-regular} if each vertex has degree $\Delta$.

\section{Our results}

Our main result is the following theorem, which provides a lower bound on the quantum query complexity of local search for any graph in terms of its congestion.
\begin{restatable}{mytheorem}{mainresult} \label{thm:main_result}
Let $G = (V, E)$ be a connected undirected graph with $n$ vertices. 
    Then the quantum query complexity of local search on $G$ is $\Omega\left(\frac{n^{\frac{3}{4}}}{\sqrt{g}}\right)$, where $g$ is the vertex congestion of the graph.
\end{restatable}

This bound is the square root of the corresponding bound of $\Omega\left( n^{1.5}/g \right)$ on the \emph{classical} randomized query complexity of local search found by \cite{branzei2024sharp}. We obtain \cref{thm:main_result}  by invoking the strong weighted adversary method on a similar construction to the one in  \cite{branzei2024sharp}.

If a graph with $n$ vertices has expansion $\beta$ and maximum degree $\Delta$, then its  vertex congestion $g$ satisfies the following inequality (see, e.g., Lemma 16 in \cite{branzei2024sharp}):
\begin{align} \label{eq:congestion_inequality}
g \in O\left( n \log^2(n) \cdot \frac{\Delta}{\beta}\right) \,.
\end{align}
    
Using \eqref{eq:congestion_inequality} in   \cref{thm:main_result} implies the following:

\begin{corollary} \label{thm:expansion}
Let $G = (V,E)$ be a connected undirected $\beta$-expander with $n$ vertices and maximum degree $\Delta$.
Then the quantum query complexity of local search on $G$ is $ \Omega\left(\frac{\sqrt{\beta} \; n^{\frac{1}{4}}}{\sqrt{\Delta} \; \log{n}} \right) \,. $ 
\end{corollary}

On constant degree expanders, \cref{thm:main_result} implies the following corollary:
\begin{corollary}
Let $G = (V, E)$ be a connected undirected $\Delta$-regular $\beta$-expander with $n$ vertices, where $\Delta$ and $\beta$ are constant. Then the quantum query complexity of local search on $G$ is $\Omega\left( \frac{n^{\frac{1}{4}}}{\sqrt{\log{n}}} \right)$.

\end{corollary}
\begin{proof}
By Corollary D.1 in \cite{branzei2024sharp}, such graphs have congestion $g \in \cO(n \log(n))$. Therefore by \cref{thm:main_result}, the quantum query complexity on every such graph is $\Omega\bigl( \frac{n^{\frac{1}{4}}}{\sqrt{\log{n}}} \bigr)$.
\end{proof}

We compare this bound to the quantum extension of Aldous' upper bound detailed in \cite{Aaronson06}, which uses Grover's search to obtain an improved quantum algorithm for local search.


\medskip  

\noindent \textbf{Theorem A.} [\cite{Aaronson06}, Theorem $3.2$] 
\emph{Let $G = (V,E)$ be a connected undirected graph with $n$ vertices and maximum degree $\Delta$. Then the quantum query complexity of local search on $G$ is $\cO\left( n^{\frac{1}{3}} \Delta^{\frac{1}{6}} \right)$.}


For constant degree expanders, this bound simplifies to $\cO(n^{\frac{1}{3}})$, leaving a polynomial gap between the upper and lower bounds for quantum local search on such graphs. This  is in contrast to the classical case where the exponent is known to be $1/2$. 
It would be interesting to close this gap.

\section{Preliminaries} \label{sec:preliminaries}

\subsection{The strong weighted adversary method} 

A powerful method for giving lower bounds in the quantum setting is the relational adversary method \cite{Ambainis02}. Several variants of the method exist, such as the strong weighted adversary \cite{Zhang05}. \cite{SS05} showed that multiple quantum relational adversary methods are in fact equivalent.

\smallskip 

We use the strong weighted adversary method from \cite{Zhang05}, formally stated next.


\smallskip  

\noindent \textbf{Theorem B.} [The strong weighted adversary method~\cite{Zhang05}]
\emph{Let $A,B$ be  finite sets and  $\mathcal{X} \subseteq B^A$  a set of functions mapping $A$ to $B$. Let $\mathcal{H} : \mathcal{X} \to \{0,1\}$ be a map that assigns to each function in $\cX$ the label  $0$ or $1$.  We are given oracle access to a function $F \in \cX$ and the problem is to compute the label $\mathcal{H}(F)$ using as few queries as possible.}

\emph{Let $r : \cX \times \cX \to \mathbbm{R}_{\ge 0}$ be a  non-zero symmetric function of our choice 
with $r(F_1,F_2) = 0$ when $\cH(F_1) = \cH(F_2)$.
Additionally, define}
\begin{itemize}
\item \emph{${r}' : \cX \times \cX \times A \to \mathbbm{R}_{\ge 0}$ such that  ${r}'(F_1,F_2,a) \cdot  {r}'(F_2,F_1,a) \ge r^2(F_1,F_2)$ for all $F_1, F_2 \in \cX$ and all $a \in A$ with $F_1(a) \ne F_2(a)$.}
\item \emph{$M : \cX \to \mathbb{R}$ such that $M(F_1) = \sum_{F_2 \in \cX} r(F_1,F_2)$ for all   $F_1 \in \cX$. } 
\item \emph{$\nu : \cX \times A \to \mathbb{R}$ such that $\nu(F_1,a) = \sum_{F_2 \in \cX} {r}'(F_1,F_2,a)$ for all   $F_1 \in \cX$ and all $a \in A$.}
\end{itemize}
\emph{Then the quantum query complexity of the problem is at least:}
\begin{align}
    \min_{\substack{F_1,F_2 \in \cX,\; a \in A:\\ r(F_1,F_2) > 0\\ F_1(a) \ne F_2(a)}} {\left( \frac{M(F_1)M(F_2)}{\nu(F_1,a)\nu(F_2,a)}\right)^{\frac{1}{2}} }\,.
\end{align}


\subsection{Setup}

Next we explain the construction from \cite{branzei2024sharp} that we use.
Since the graph $G$ has vertex congestion $g$, there is an all-pairs set of paths $\mathcal{P} = \left\{P^{u,v}\right\}_{u,v \in [n]}$ with vertex congestion $g$ \footnote{That is, each vertex is used at most $g$ times across all the paths.}.  For each $u,v \in V$, let $dist(u,v)$ be the length of the shortest path from $u$ to $v$.


Let $L \in [n]$, with $L \geq 2$, be a parameter that we set later. Given a sequence of $k$ vertices $\X = (x_1, \ldots, x_k)$, we write $\X_{1 \to j}=(x_1,\ldots,x_j)$ to refer to a prefix of the sequence, for an index $j \in [k]$.
Given a walk $Q = (v_1, \ldots, v_k)$ in $G$, let $Q_i$ refer to the $i$-th vertex in the walk  (i.e. $Q_i = v_i$). For each vertex $u \in [n]$, let $\multiplicity(Q,u)$ be the number of times that vertex $u$ appears in $Q$.

\begin{mydefinition}[Staircase]
\label{def:staircase}
Given a sequence  $\X = (x_1, \ldots, x_{k})$ of vertices in  $G$, a \emph{staircase} induced by $\X$ is a walk $S_{\X} = S_{\X,1} \circ \ldots \circ S_{\X, k-1}$, where each $S_{\X, i}$ is a path in $G$ starting at vertex $x_{i}$ and ending at $x_{i+1}$. 
Each vertex $x_i$ is called a \emph{milestone} and each path $S_{\X,i}$ a \emph{quasi-segment}.

The staircase $S_{\vec{x}}$ is said to be \emph{induced by   $\vec{x}$ and  $\mathcal{P} = \left\{P^{u,v}\right\}_{u,v \in [n]}$} if additionally  $S_{\X,i} = P^{x_i, x_{i+1}}$ for all $i \in [k-1]$.
\end{mydefinition}

 \begin{mydefinition}[Tail of a staircase] \label{def:tail}
    Let $S_{\vec{x}} = S_{\vec{x},1} \circ \ldots \circ S_{\vec{x}, k-1}$  be a staircase induced  by some sequence $\X  = (x_1, \ldots, x_k) \in [n]^k$.
    For each $j \in [k-1]$, let $T = S_{\vec{x},j} \circ \ldots \circ S_{\vec{x}, k-1}$.
    Then 
   $Tail(j, S_{\X})$ is obtained from $T$ by removing the first occurrence of $x_j$ in $T$ (and only the first occurrence).
   We also define $Tail(k, S_{\X})$  to be the empty sequence. 
\end{mydefinition}

Next we define a set of functions $\mathcal{X}$.

\begin{mydefinition}[The functions $f_{\vec{x}}$ and $g_{\vec{x},b}$; the  set $\mathcal{X}$]
\label{def:mathcal_X}
 Suppose $\mathcal{P} =  \{P^{u,v} \}_{u,v \in [n]} $ is an all-pairs set  of  paths in $G$. For each sequence  of vertices  $\X \in \{1\} \times [n]^{L}$, define a function \\ $f_{\X} : [n] \to  \{-n^2-n, \ldots, 0,\ldots,n\}$  such that for each $v \in [n]$:
    \begin{itemize}
        \item If $v \notin S_{\X}$, then  set $f_{\X}(v) = dist(v,1)$, where $S_{\X}$ is the staircase induced by $\X$ and $\cP$.
        \item If $v \in S_{\X}$, then set $f_{\X}(v) = -i\cdot n - j$, where $i$ is the maximum index with $v \in P^{x_i,x_{i+1}}$, and $v$ is the $j$-th vertex in $P^{x_i,x_{i+1}}$.
    \end{itemize}
Also, for each $\X \in \{1\} \times [n]^L$ and $b \in \{0,1\}$, let $g_{\X,b} : [n] \to \{-n^2-n, \ldots, 0,\ldots,n\} \times \{-1,0,1\}$ be such that, for all $v \in [n]$:
\begin{align} 
    g_{\X,b}(v) = \begin{cases}
            \bigl(f_{\X}(v), b\bigr) & \text{if} \; v = x_{L+1}\\
            \bigl(f_{\X}(v), -1\bigr) & \text{if}  \; v \ne x_{L+1}
        \end{cases} \,. \label{eq:def:g_X_b}
\end{align} 
Let $\mathcal{X} = \Bigl\{g_{\X,b} \mid  \X \in \{1\} \times [n]^{L} \mbox{ and }  b \in \{0,1\}\Bigr\}$.
\end{mydefinition}


For each $\vec{x} \in \{1\} \times [n]^L$, the associated 
 function $f_{\vec{x}}$ has a unique local minimum, at $x_{L+1}$ (the proof follows by   Lemma 6 and Lemma 7 in the full version\footnote{The full version of \cite{branzei2024sharp} can be found at \url{https://arxiv.org/pdf/2305.08269}.} of \cite{branzei2024sharp}).

The strong weighted adversary method gives lower bounds on the quantum query complexity of Boolean functions, so we turn the local search question into a decision problem by defining a map $\mathcal{H}$ that assigns a bit to each function in $\mathcal{X}$. 

 \begin{mydefinition}[The map $\mathcal{H}$] \label{def:map_H}
 Suppose $\mathcal{P} =  \{P^{u,v} \}_{u,v \in [n]} $ is an all-pairs set  of  paths in $G$ and $\mathcal{X}$ is the set of functions $g_{\vec{x},b}$ from Definition~\ref{def:mathcal_X}. Define $\mathcal{H} : \mathcal{X} \to \{0,1\}$ as 
 \[
\mathcal{H}(g_{\vec{x},b}) = b \qquad  \forall \vec{x} \in \{1\} \times [n]^{L} \; \text{and} \; b \in \{0,1\} \,. 
\]
 \end{mydefinition}

Thus the decision problem is: given a graph $G$  and oracle access to a function $g_{\vec{x},b} \in \mathcal{X}$, return the value $\mathcal{H}(g_{\vec{x},b}) =b$. This means: find the hidden bit, which only exists at the vertex representing the  local minimum of $f_{\vec{x}}$.
Measuring the query complexity of this decision problem will give  the answer for local search, as the next two problems have query complexity within additive $1$:

\smallskip 

$\bullet \; \;$ \emph{search problem:} given oracle access to a function $f_{\vec{x}}$,  find a vertex $v$ that is a local minimum; 

\smallskip 

$\bullet \; \;$ \emph{decision problem:}  given oracle access to the  function $g_{\vec{x},b}$, find $\mathcal{H}(g_{\vec{x},b})$. 

\section{Proof of Theorem \ref{thm:main_result}}

In this section we prove Theorem~\ref{thm:main_result}.

\mainresult*
\begin{proof}
Let $L = \sqrt{n}$. The plan is to invoke Theorem B using the set of functions  $\mathcal{X}$ from \cref{def:mathcal_X}, the map $\mathcal{H}$ from \cref{def:map_H}, and the functions (aka relations) $r, r'$ that we define next.

By definition, for each function  $g_{\vec{x},b} \in \mathcal{X}$, its underlying  staircase $S_{\vec{x}}$ has $L+1$ milestones and $L$ quasi-segments.
For all sequences of milestones $\X,\Y \in \{1\} \times [n]^L$, let $J_{\X,\Y} = \max\{i : \X_{1 \to i} = \Y_{1 \to i} \}$; i.e. the number of milestones in the shared prefix of $\X$ and $\Y$. 

We borrow the relation function defined in \cite{branzei2024sharp} as a component $r^*$ of our relation  $r$. First, we say that a sequence of vertices $\X = (x_1, \ldots, x_{L+1})$ is  \emph{good} if  $x_i \ne x_j$ for all $i,j \in [L+1]$ with $ i < j$; otherwise, $\X$ is \emph{bad}.
That is, a good $\X$ only involves distinct milestones.
For each bit $b \in \{0,1\}$, the function $g_{\X,b} \in \mathcal{X}$ is \emph{good} if $\X$ is good, and \emph{bad} otherwise.

Let $r^* : \mathcal{X} \times \mathcal{X} \to \mathbb{R}_{\ge 0}$  be a symmetric function such that for each \mbox{$\X, \Y \in \{1\} \times [n]^L$} and $b_1,b_2 \in \{0,1\}$, we have:
\[
    r^*(g_{\X,b_1}, g_{\Y,b_2}) = \begin{cases}
        0 & \text{ if at least one of the following holds: } b_1 = b_2 \text{ or } \X \text{ is bad or } \Y \text{ is bad.}\\
        n^j & \text{ otherwise, where } j \text{ is the maximum index for which } \X_{1 \to j} = \Y_{1 \to j}\,. 
    \end{cases}
\]

Let $r : \mathcal{X} \times \mathcal{X} \to \mathbb{R}_{\ge 0}$ be:
\begin{align}
    r(g_{\X,b_1},g_{\Y,b_2}) = \begin{cases}
        0 & \text{if } \X = \Y\\
        r^*(g_{\X,b_1},g_{\Y,b_2}) & \text{otherwise}
    \end{cases}
\end{align}

Next we define  $r' : \mathcal{X} \times \mathcal{X} \times [n] \to \mathbb{R}_{\ge 0}$. For each pair of functions  $g_{\X,b_1},g_{\Y,b_2} \in \cX$ and  each vertex $v \in [n]$, let 
\begin{align*}
    r'(g_{\X,b_1},g_{\Y,b_2},v) = \begin{cases}
        0 & \text{if } g_{\X,b_1}(v) = g_{\Y,b_2}(v) \text{ or } \cH(g_{\X,b_1}) = \cH(g_{\Y,b_2})\\
        r(g_{\X,b_1},g_{\Y,b_2})\cdot g/n^{1.5} &
        \text{otherwise, if }
        v \in Tail(J_{\X,\Y}, S_{\X}) \text{ and } v \notin Tail(J_{\X,\Y},S_{\Y})\\
        r(g_{\X,b_1},g_{\Y,b_2}) \cdot n^{1.5}/g &
        \text{otherwise, if }
        v \in Tail(J_{\X,\Y}, S_{\Y}) \text{ and } v \notin Tail(J_{\X,\Y},S_{\X})\\
        r(g_{\X,b_1},g_{\Y,b_2}) & \text{otherwise}
    \end{cases}
\end{align*}
The functions $r$ and $r'$ constitute a weight scheme for use with the strong weighted adversary method.
We can now invoke Theorem B with parameters $\mathcal{X}$, $\mathcal{H}$, $r$, and $r'$, resulting in a quantum lower bound of 
\begin{align} \label{eq:lb_obtained_abstract}
    \min_{\substack{F_1,F_2 \in \cX,\; a \in [n]:\\ r(F_1,F_2) > 0\\ F_1(a) \ne F_2(a)}} {\left( \frac{M(F_1)M(F_2)}{\nu(F_1,a)\nu(F_2,a)}\right)^{\frac{1}{2}} },
\end{align}
where $M$ and $\nu$ are functions of $r$ and $r'$ as defined in Theorem B. Next we estimate \eqref{eq:lb_obtained_abstract}.

Take an arbitrary choice of $g_{\X,b_1},g_{\Y,b_2} \in \cX$ and $v \in V$ such that $r(g_{\X,b_1},g_{\Y,b_2}) > 0$ and \mbox{$g_{\X,b_1}(v) \ne g_{\Y,b_2}(v)$}.
Since $r(g_{\X,b_1},g_{\Y,b_2}) > 0$, we have that both $g_{\X,b_1}$ and $g_{\Y,b_2}$ are good.
Therefore by \cref{lem:M_lb} we have 
\begin{align} \label{eq:lb_M_g_x_b_and_g_y_b}
M(g_{\X,b_1}) \ge \frac{1}{2e} L n^{L+1} \; \; \mbox{and} \; \; M(g_{\Y,b_2}) \ge \frac{1}{2e} L n^{L+1} \,.
\end{align}

We can now bound $\nu$ by considering two cases.

\begin{description}
\item[{Case 1: $v \in Tail(J_{\X,\Y},S_{\X})$.}] By definition of $\nu(g_{\X,b_1},v)$ and $\mathcal{X}$, we have 
\begin{align} \label{eq:case_1_using_def_nu}
    \nu(g_{\X,b_1},v) &=  \sum_{{g_{\Y,b_2} \in \cX}} r'(g_{\X,b_1},g_{\Y,b_2},v) = \sum_{j = 1}^{L+1} \sum_{\substack{g_{\Y,b_2} \in \cX:\\J_{\X,\Y} = j}} r'(g_{\X,b_1},g_{\Y,b_2},v)\,.
\end{align}

Since the family of functions $\mathcal{X}$ contains exactly one function $g_{\Y,b_2}$ with $J_{\X,\Y} = L+1$ and $b_1 \ne b_2$, we can continue from  \eqref{eq:case_1_using_def_nu} by rewriting $\nu(g_{\X,b_1},v)$ as 
    \begin{align}
           \nu(g_{\X,b_1},v) &= n^{L+1} + \sum_{j=1}^L \sum_{\substack{g_{\Y,b_2} \in \cX:\\J_{\X,\Y} = j}} r'(g_{\X,b_1},g_{\Y,b_2},v) \notag \\
           & = n^{L+1} + \sum_{j=1}^L \sum_{\substack{g_{\Y,b_2} \in \cX:\\J_{\X,\Y} = j\\v \in Tail(j,\, S_{\Y})}} n^j + \sum_{j=1}^L \sum_{\substack{g_{\Y,b_2} \in \cX:\\J_{\X,\Y} = j\\v \notin Tail(j,\, S_{\Y})}} n^j \cdot g/n^{1.5} \,. \label{eq:case_1_rewriting_nu_further}
    \end{align} 

Lemma~\ref{lem:sum_first_part_both_cases} gives 
\begin{align} \label{eq:case_1_ub_first_part}
 \sum_{j=1}^L   \sum_{\substack{g_{\Y,b_2} \in \cX:\\J_{\X,\Y} = j\\v \in Tail(j,\, S_{\Y})}} n^j \leq g n^L + L^2 g n^{L-1}  \,. 
\end{align}

Meanwhile,  
\begin{align}
    \big| \left\{ \vec{y} \in \{1\} \times [n]^L \mid J_{\vec{x}, \vec{y}} = j, v \not \in Tail(j, S_{\vec{y}}) \right\} \bigr| \leq  \big| \left\{ \vec{y} \in \{1\} \times [n]^L \mid J_{\vec{x}, \vec{y}} = j \right\} \bigr|  \leq  n^{L+1-j},
\end{align}
and so 
\begin{align} \label{eq:case_1_ub_second_part}
    \sum_{\substack{g_{\Y,b_2} \in \cX:\\J_{\X,\Y} = j\\v \notin Tail(j,\, S_{\Y})}} n^j \cdot g/n^{1.5} \leq n^{L+1-j} \cdot n^j \cdot g/n^{1.5} = n^{L+1} \cdot g/n^{1.5} \,.
\end{align}

Using inequalities \eqref{eq:case_1_ub_first_part} and \eqref{eq:case_1_ub_second_part} in \eqref{eq:case_1_rewriting_nu_further}, we obtain 
\begin{align}
\nu(g_{\X,b_1},v) &\le n^{L+1} +
                \bigl( g n^L + L^2 g n^{L-1} \bigr) +
                \sum_{j=1}^L n^{L+1} \cdot g/n^{1.5} \notag \\
                & = n^{L+1} + g n^L + L^2 g n^{L-1} + L n^{L+1} g/n^{1.5}  \,. \label{eq:case_1_almost_final_bound_nu}
    \end{align}
 Since $g \in [n, n^2]$ for any graph $G$ and $L = \sqrt{n}$, inequality \eqref{eq:case_1_almost_final_bound_nu} gives  
\begin{align} \label{eq:case_1_final_ub_nu}
            \nu(g_{\X,b_1},v) &\le 4gn^L \,.
    \end{align}
\item[{Case 2: $v \notin Tail(J_{\X,\Y}, S_{\X})$.}] By definition of $\nu(g_{\X,b_1},v)$ and $\mathcal{X}$, we have 
\begin{align} \label{eq:case_2_nu_initial}
\nu(g_{\X,b_1},v) &= \sum_{j = 1}^{L+1} \sum_{\substack{g_{\Y,b_2} \in \cX:\\J_{\X,\Y} = j}} r'(g_{\X,b_1},g_{\Y,b_2},v) \,.
\end{align}
Since there is exactly one function $g_{\Y,b_2} \in \mathcal{X}$ with $J_{\X,\Y} = L+1$ and $b_1 \ne b_2$, we can further decompose the right hand side of \eqref{eq:case_2_nu_initial} as 
\begin{align} \label{eq:case_2_nu_second}
    \nu(g_{\X,b_1},v) & = n^{L+1} + \sum_{j=1}^L \sum_{\substack{g_{\Y,b_2} \in \cX:\\J_{\X,\Y} = j}} r'(g_{\X,b_1},g_{\Y,b_2},v) \,.
\end{align}
Since we chose $g_{\X,b_1}, g_{\Y,b_2}, v$ such that $g_{\X,b_1}(v) \ne g_{\Y,b_2}(v)$, and we are in the case $v \notin Tail(J_{\X,\Y}, S_{\X}$), it follows that $v \in Tail(J_{\X,\Y}, S_{\Y})$. Therefore, we can rewrite  \eqref{eq:case_2_nu_second} as 
    \begin{align}
        \nu(g_{\X,b_1},v) 
            &= n^{L+1} + \sum_{j=1}^L \sum_{\substack{g_{\Y,b_2} \in \cX:\\J_{\X,\Y} = j\\v \in Tail(j,\, S_{\Y})}} n^j\cdot n^{1.5}/g \,. \label{eq:case_2_some_identity} 
    \end{align}
Using     Lemma~\ref{lem:sum_first_part_both_cases} in  \eqref{eq:case_2_some_identity}  gives 
\begin{align}
         \nu(g_{\X,b_1},v)   &\le n^{L+1} + \left( gn^L  + L^2 g n^{L-1} \right) \cdot n^{1.5}/g  \,. 
                \label{eq:case_2_first_ub}
\end{align}
Since $L = \sqrt{n}$, inequality \eqref{eq:case_2_first_ub} gives 
\begin{align} \label{eq:case_2_final_ub_nu}
    \nu(g_{\X,b_1},v)  \leq 3n^{L+1.5} \,.
\end{align}
This completes case 2.
\end{description}
Combining inequality \eqref{eq:case_1_final_ub_nu} from case 1 and inequality \eqref{eq:case_2_final_ub_nu} from case 2, we obtain 
\begin{align} \label{eq:nu_upper_bound}
    \nu(g_{\X,b_1},v)  \leq \max \left\{ 4gn^L, 3n^{L+1.5}  \right\} \,.
\end{align}

Because $g_{\X,b_1}(v) \ne g_{\Y,b_2}(v)$ and $r(g_{\X,b_1},g_{\Y,b_2}) > 0$, we have $v \in Tail(J_{X,Y}, S_{\X}) \cup Tail(J_{X,Y}, S_{\Y})$. Then:
\begin{enumerate}[(a)]
\item If $v \in Tail(J_{X,Y}, S_{\X})$,  by case $1$ applied to $\vec{x}$ we have $\nu(g_{\X,b_1},v)  \leq 4gn^L$.
\item If $v \in Tail(J_{X,Y}, S_{\Y})$,  by case $1$ applied to $\vec{y}$ we have $\nu(g_{\Y,b_2},v)  \leq 4gn^L$. 
\end{enumerate}
Combining (a-b) with \eqref{eq:nu_upper_bound} gives 
\begin{align} \label{eq:nu_product_ineq}
\nu(g_{\X,b_1},v)  \nu(g_{\Y,b_2},v) \leq 4gn^L \cdot \max \left\{ 4gn^L, 3n^{L+1.5}  \right\} \,. 
\end{align}
Then  inequalities \eqref{eq:lb_M_g_x_b_and_g_y_b} and  \eqref{eq:nu_product_ineq} imply 
\begin{align}
    \sqrt{\frac{M(g_{\X,b_1})M(g_{\Y,b_2})}{\nu(g_{\X,b_1},v) \nu(g_{\Y,b_2},v) }} &\ge \sqrt{\frac{\frac{1}{4e^2} L^2 n^{2L+2}}{4gn^L\cdot \max\{4gn^L, 3n^{L+1.5}\}}}\\
    &\ge \sqrt{\frac{1}{64e^2} \cdot \min\left\{\frac{n^3}{g^2}, \frac{n^{1.5}}{g}\right\} }\\
    &= \frac{1}{8e} \cdot \min\left\{\frac{n^{1.5}}{g}, \frac{n^{0.75}}{\sqrt{g}}\right\} 
    \label{eq:almost_final_bound_overall}
\end{align}
If $g \ge n^{1.5}$, then the  bound claimed in the theorem statement is less than or equal to $\Omega(1)$, which is trivially true.
Otherwise, $g < n^{1.5}$, in which case the bound in \eqref{eq:almost_final_bound_overall} simplifies to:
\begin{align}
    \sqrt{\frac{M(g_{\X,b_1})M(g_{\Y,b_2})}{\nu(g_{\X,b_1},v) \nu(g_{\Y,b_2},v) }} \ge \frac{1}{8e} \cdot \frac{n^{0.75}}{\sqrt{g}} \label{eq:final_bound_this_proof}
\end{align}
Since \eqref{eq:final_bound_this_proof} holds for arbitrary choices of $g_{\X,b_1}, g_{\Y,b_2}, v$ such that $r(g_{\X,b_1} , g_{\Y,b_2}) > 0$ and \mbox{$g_{\X,b_1}(v) \ne g_{\Y,b_2}(v)$},  Theorem B implies that the quantum query complexity of local search is $\Omega\left(\frac{n^{0.75}}{\sqrt{g}}\right)$ as required.
\end{proof}

\begin{lemma} \label{lem:sum_first_part_both_cases}
In the setting of Theorem 1, let $g_{\X,b_1},g_{\Y,b_2} \in \cX$ and $v \in V$ be such that $r(g_{\X,b_1},g_{\Y,b_2}) > 0$ and \mbox{$g_{\X,b_1}(v) \ne g_{\Y,b_2}(v)$}. Then 
    \begin{align} \label{eq:sum_first_part_both_cases}
  \sum_{j=1}^L  \sum_{\substack{g_{\Y,b_2} \in \cX:\\J_{\X,\Y} = j\\v \in Tail(j,\, S_{\Y})}} n^j \leq  g n^L + L^2 g n^{L-1} \,. 
\end{align}
\end{lemma}
\begin{proof}
Since the graph $G$ has vertex congestion $g$, there is an all-pairs set of paths 
$\mathcal{P} = \left\{P^{u,v}\right\}_{u,v \in [n]}$ with vertex congestion $g$.
For each $u,v \in [n]$, let $\numPaths{v}(u)$ be the number of paths in $\mathcal{P}$ that start at vertex $u$ and contain $v$:
    \begin{equation} \label{eq:numPaths}
        \numPaths{v}(u) = \Bigl|\{P^{u,w} \in \cP : w \in [n], v \in P^{u,w}\} \Bigr|  \,. 
    \end{equation}

For each $j \in [L+1]$, let 
\begin{align}
T_j = \left\{\Y \in \{1\} \times [n]^L \mid \max\{i : \X_{1 \to i} = \Y_{1 \to i} \} = j \text{ and } v \in Tail(j,S_{\Y}) \right\} \,.
\end{align}
By Lemma 12 in \cite{branzei2024sharp}, we have  
\begin{align} \label{eq:ub_T_j_size}
|T_j| \leq \numPaths{v}(x_j) \cdot n^{L - j} + L \cdot g \cdot n^{L - j - 1} \,.
\end{align}
Then for each $j \in [L]$,
\begin{align} \label{eq:common_case_initial_ineq}
\sum_{\substack{g_{\Y,b_2} \in \cX:\\J_{\X,\Y} = j\\v \in Tail(j,\, S_{\Y})}} n^j = |T_j| \cdot n^j \leq \left(\numPaths{v}(x_j) \cdot n^{L - j} + L \cdot g \cdot n^{L - j - 1} \right) n^j \,.
\end{align}
Since $r(g_{\vec{x}, b_1}, g_{\vec{y}, b_2}) > 0$, both $\vec{x}$ and $\vec{y}$ are good. Since $\vec{x}$ is good, it does not have repeated vertices, and so $\sum_{j=1}^L \mathfrak{q}_v(x_j) \leq g$, which combined with  \eqref{eq:common_case_initial_ineq} gives 
\begin{align}
\sum_{j=1}^L  \sum_{\substack{g_{\Y,b_2} \in \cX:\\J_{\X,\Y} = j\\v \in Tail(j,\, S_{\Y})}} n^j = \sum_{j=1}^L  |T_j| \cdot n^j \leq \sum_{j=1}^L \left(\numPaths{v}(x_j) \cdot n^{L - j} + L \cdot g \cdot n^{L - j - 1} \right) n^j \leq  g n^L + L^2 g n^{L-1} 
 \,.
\end{align}
This completes the proof.
\end{proof}

\begin{lemma} \label{lem:M_lb}
    Let $g_{\X,b_1} \in \cX$. If $g_{\X,b_1}$ is good, then $\sum_{g_{\Y,b_2} \in \cX} r(g_{\X,b_1},g_{\Y,b_2}) \ge \frac{1}{2e}\cdot L \cdot n^{L+1}$.
\end{lemma}
\begin{proof}
    The bulk of the work has already been done by \cref{lem:r_star_M_X_large}. 
    However, that lemma used $r^*$ whereas our function $r$ has value $0$ when the milestones are equal.
    Here we bridge that gap as follows:
    \begin{align*}
        \sum_{g_{\Y,b_2} \in \cX} r(g_{\X,b_1},g_{\Y,b_2}) &\ge 
            \sum_{g_{\Y,b_2} \in \cX} r^*(g_{\X,b_1},g_{\Y,b_2})
            - \sum_{\substack{g_{\Y,b_2} \in \cX\\ \Y = \X}} r^*(g_{\X,b_1},g_{\Y,b_2}) \explain{By definition of $r$.}\\
        &\ge \frac{1}{2e} \cdot (L+1) \cdot n^{L+1}
            - \sum_{\substack{g_{\Y,b_2} \in \cX\\ \Y = \X}} r^*(g_{\X,b_1},g_{\Y,b_2}) \explain{By \cref{lem:r_star_M_X_large}}\\
        &\ge \frac{1}{2e} \cdot (L+1) \cdot n^{L+1} - n^{L+1} \explain{By definition of $r^*$, since exactly one $\Y = \X$.}\\
        &\ge \frac{1}{2e} \cdot L \cdot n^{L+1} \,.
    \end{align*}
    This completes the proof.
\end{proof}

\section{Concluding Remarks}

There is an intriguing gap between the lower bound  $\Omega\left( \frac{n^{\frac{1}{4}}}{\sqrt{\log(n)}} \right)$ and upper bound  $O\left( n^{\frac{1}{3}} \right)$ for constant degree expanders. While the construction from \cite{branzei2024sharp} is optimal in the classical setting, no more efficient algorithm is known for expanders in the quantum case.

It remains an open question to understand the connection between the classical and quantum query complexity of local search. Existing results largely adapt classical constructions by finding suitable weight schemes for the quantum setting. Is there a natural relationship such that an improvement in one automatically translates to the other?

\bibliographystyle{alpha}
\bibliography{quantum_bib}

\appendix

\section{Additional facts  from prior work}

\begin{lemma} [\cite{branzei2024sharp}, Lemma 14] \label{lem:r_star_M_X_large}
Let $F_1 \in \mathcal{X}$. If $F_1$ is good, then 
\[ 
\sum_{F_2 \in \cX} r^*(F_1,F_2) \geq \frac{1}{2e} \cdot (L+1) \cdot n^{L+1}\,.
\]
\end{lemma}


\end{document}